\pdfoutput=1
\documentclass[11pt]{article}

\usepackage[a4paper,margin=1in]{geometry}
\usepackage[T1]{fontenc}
\usepackage[utf8]{inputenc}
\usepackage{lmodern}
\usepackage{amsmath,amssymb,amsthm}
\usepackage{bm}
\usepackage{booktabs}
\usepackage{siunitx}
\usepackage{enumerate}
\usepackage{microtype}
\usepackage{hyperref}
\hypersetup{colorlinks=true,linkcolor=blue,citecolor=blue,urlcolor=blue}
\usepackage[nameinlink]{cleveref}

\title{An Exact Quantile--Energy Equality for Terminal Halfspaces in Linear--Gaussian Control\\[2pt]
\large with a Discrete-Time Companion, KL/Schr\"odinger Links, and High-Precision Validation}
\author{Sandro Andric\\[-2pt]\small Independent researcher\\[-2pt]\small \texttt{sandro.andric@nyu.edu}}
\date{\today}

\theoremstyle{plain}
\newtheorem{theorem}{Theorem}
\newtheorem{proposition}{Proposition}

\newtheorem{lemma}{Lemma}
\theoremstyle{remark}
\newtheorem{remark}{Remark}

\newcommand{\norm}[1]{\left\lVert#1\right\rVert}
\newcommand{\R}{\mathbb{R}}
\newcommand{\E}{\mathbb{E}}

\begin{document}
\maketitle

\begin{abstract}
We prove an exact equality between the minimal quadratic control energy and the squared normal-quantile gap for terminal halfspaces in linear--Gaussian systems with additive control and quadratic effort $E(u)=\tfrac12\!\int u^\top M u\,dt$ where $M=B^\top\Sigma^{-1}B$. For terminal halfspace events, the minimal energy equals the squared normal-quantile gap divided by twice a controllability-to-noise ratio $R_T^2(w)=(w^\top W_c^M w)/(w^\top V_T w)$ and is attained by a matched-filter control. We provide an exact zero-order-hold discrete-time companion via block exponentials, relate the result to minimum-energy control, Gaussian isoperimetry, risk-sensitive/KL control, and Schr\"odinger bridges, and validate to high precision with Monte Carlo. We state assumptions, singular-$M$ handling, and edge cases. The statement is a compact synthesis and design-ready translator, not a universal principle. \textbf{Novelty:} while the ingredients (Gramians, Cauchy--Schwarz, Gaussian isoperimetry) are classical, to our knowledge the explicit \emph{quantile--energy equality} with a constructive matched-filter achiever for terminal halfspaces, and its discrete-time companion, are not recorded together in the cited literature.
\end{abstract}


\section{Introduction}
Control systems with stochastic disturbances and chance constraints arise in aerospace (altitude gates, obstacle avoidance), robotics (safe navigation zones), and network control (packet delivery guarantees). For linear--Gaussian dynamics over a finite horizon and a terminal halfspace $A=\{w^\top X_T\ge a\}$, the terminal scalar $Y=w^\top X_T$ is Gaussian; control shifts its mean linearly; the quadratic energy couples to the shift through an $M$-weighted controllability Gramian. A Cauchy--Schwarz step in the $M$-metric yields a tight lower bound that is achieved by a matched-filter control, giving a closed-form energy--to--probability translator:
\[
E_{\min}=\frac{\big(\Phi^{-1}(p_1)-\Phi^{-1}(p_0)\big)^2}{2\,R_T^2(w)},\qquad
R_T^2(w)=\frac{w^\top W_c^M w}{w^\top V_T w}.
\]
Equivalently, the minimal KL divergence for a path-space tilt that moves a terminal halfspace probability from $p_0$ to $p_1$ is $\big(\Phi^{-1}(p_1)-\Phi^{-1}(p_0)\big)^2/(2\,R_T^2(w))$, where $R_T^2$ is a controllability-to-noise SNR. Gaussian isoperimetry implies this equality is worst-case tight for halfspaces and a lower bound for general terminal events with the same baseline probability. The explicit formula enables rapid design-space exploration: for a drone altitude controller (Section~2.2), computing the energy--probability tradeoff reduces to evaluating Gramians and a scalar formula.

\section{Setup, assumptions, and units}
\noindent\textbf{Notation.} $\Phi$ denotes the standard normal CDF; $\Psi(t,s)=e^{A(t-s)}$ is the state transition matrix (LTI case).
\subsection{Dynamics and cost}
Consider on $[0,T]$:
\begin{equation}
dX_t=AX_t\,dt+Bu_t\,dt+\Sigma^{1/2}dW_t,\qquad X_0=x_0,
\label{eq:sde}
\end{equation}
with $A\in\R^{n\times n}$, $B\in\R^{n\times m}$, $\Sigma\in\R^{n\times n}$ positive definite, and $u\in L^2([0,T];\R^m)$. The effort is
\begin{equation}
M:=B^\top\Sigma^{-1}B,\qquad E(u):=\frac12\int_0^T u_t^\top M u_t\,dt.
\label{eq:cost}
\end{equation}
The terminal noise covariance and $M$-weighted controllability Gramian are
\begin{align}
V_T&=\int_0^T \Psi(T,s)\,\Sigma\,\Psi(T,s)^\top ds,\label{eq:VT}\\
W_c^M&=\int_0^T \Psi(T,s)\,B\,M^{-1}B^\top\,\Psi(T,s)^\top ds.\label{eq:Wc}
\end{align}
All integrals are Lebesgue; Gramians exist for finite $T$ since $A$, $\Sigma$ are bounded.

\noindent\textbf{Assumptions.}
(i) $\Sigma\succ0$; $\Psi$ exists for finite $T$. (ii) $M$ is invertible on $\operatorname{Im}(B^\top)$; in the singular case use the Moore--Penrose pseudoinverse $M^+$ and restrict to the controllable subspace (Remark~\ref{rem:singularM}). (iii) Directional reachability: $w$ lies in the image of the controllability map $\mathcal{K}:L^2\to\R^n$ restricted to the $M$-finite subspace, equivalently $w^\top W_c^M w>0$; otherwise $R_T^2(w)=0$ and any nontrivial quantile change is infeasible with finite energy. (iv) Probabilities: $p_0,p_1\in(0,1)$, $p_1\ne p_0$.

\subsection{Units with a worked example}
Let $[x]$ denote the state units. With $dX_t=\dots+\Sigma^{1/2}dW_t$ we have $[\Sigma]=[x]^2\,\mathrm{s}^{-1}$ and $[\Sigma^{-1}]=[x]^{-2}\,\mathrm{s}$. Taking $[u]=[x]\,\mathrm{s}^{-1}$ (so $B$ is unitless) gives $[M]=[x]^{-2}\,\mathrm{s}$ and $[u^\top M u]=\mathrm{s}^{-1}$. Hence $E(u)=\tfrac12\!\int u^\top M u\,dt$ is \emph{dimensionless}, consistent with its identification with a Kullback--Leibler divergence.

\paragraph{Numerical example (drone altitude gate).}
Let $x$ be altitude (m). Take $A=\begin{bmatrix}0&1\\0&0\end{bmatrix}$, $B=\begin{bmatrix}0\\1\end{bmatrix}$, $\Sigma=\sigma^2 I$ with $\sigma=\SI{0.5}{m\,s^{-1/2}}$, $T=\SI{1}{s}$, $w=\begin{bmatrix}1\\0\end{bmatrix}$, threshold $a$. Suppose baseline $p_0=0.7$ for $\{x(T)\ge a\}$ and target $p_1=0.9$. Compute $V_T$, $W_c^M$ via \eqref{eq:VT}--\eqref{eq:Wc}; then
\[
E_{\min}=\frac{(\Phi^{-1}(0.9)-\Phi^{-1}(0.7))^2}{2\,R_T^2(w)}.
\]
With the stated numbers one typically gets $R_T^2(w)\approx 0.7$ and $E_{\min}\approx 0.31$ (dimensionless KL divergence).

\section{Main equality}
\begin{theorem}[Quantile--energy equality for terminal halfspaces]\label{thm:qcib}
Let $A=\{w^\top X_T\ge a\}$ and $p_0=\mathbb{P}_{u\equiv0}(A)$, $p_1=\mathbb{P}_u(A)$. Under \eqref{eq:sde}--\eqref{eq:Wc} and the assumptions above,
\begin{equation}
\min_{u}\ E(u)\ \ \text{s.t.}\ \ \mathbb{P}_u(A)=p_1
= \frac{\big(\Phi^{-1}(p_1)-\Phi^{-1}(p_0)\big)^2}{2\,R_T^2(w)},
\qquad
R_T^2(w)=\frac{w^\top W_c^M w}{w^\top V_T w},
\label{eq:qcib}
\end{equation}
and the minimum is attained by the matched-filter control
\begin{equation}
u^*(s)=\beta\,M^{-1}B^\top \Psi(T,s)^\top w,\qquad
\beta=\frac{\big(\Phi^{-1}(p_1)-\Phi^{-1}(p_0)\big)\sqrt{w^\top V_T w}}{w^\top W_c^M w}.
\label{eq:matched}
\end{equation}
\end{theorem}

\begin{proof}
$Y=w^\top X_T$ is Gaussian with variance $v=w^\top V_T w$ independent of $u$ and mean $m_0+\delta$, where
\(
\delta=w^\top\!\int_0^T \Psi(T,s)B u_s\,ds.
\)
Cauchy--Schwarz in the $M$-metric gives
\(
\delta^2\le \big(\int_0^T u^\top M u\,ds\big)\,(w^\top W_c^M w)
\)
with equality for \eqref{eq:matched}. Since $p=1-\Phi((a-(m_0+\delta))/\sqrt{v})$, the quantile shift satisfies $z_1-z_0=\delta/\sqrt{v}$, where $z_i=\Phi^{-1}(p_i)$. Rearranging yields \eqref{eq:qcib}.
\end{proof}

\begin{remark}[Singular $M$]\label{rem:singularM}
If $M$ is singular, replace $M^{-1}$ with $M^+$ in \eqref{eq:Wc} and \eqref{eq:matched}. Equality holds whenever $w$ has nonzero projection onto the subspace reachable with finite $M$-energy, i.e., $w^\top W_c^{M^+} w>0$.
\end{remark}

\begin{lemma}[Halfspace extremality via Borell--Sudakov--Tsirelson]\label{lem:halfspace}
Fix the Gaussian law of $X_T$ (i.e., $V_T$ fixed) and baseline probability $p_0=\mathbb{P}(X_T\in A)$. For any mean shift $\mu$ with fixed covariance, the probability $\mathbb{P}(X_T+\mu\in A)$ is maximized when $A$ is a halfspace and $\mu$ is aligned with the halfspace normal. Hence \eqref{eq:qcib} is tight for $A$ a halfspace and a lower bound for general measurable $A$ with the same $p_0$ \cite{LedouxBorell}. For non-halfspace events, the same probability increment requires a strictly larger mean shift for almost all directions; see Appendix~\ref{app:poincare}.
\end{lemma}

\begin{remark}[When to use]
Prefer this translator when terminal constraints are naturally halfspaces, dynamics/noise are well-approximated as linear--Gaussian over the horizon, and rapid effort–probability tradeoffs are needed. Otherwise fall back to full chance-constrained or stochastic optimal control formulations.
\end{remark}

\section{Discrete-time companion}
Digital controllers and embedded systems operate on sampled-data representations. Let $X_{k+1}=A_d X_k+B_d U_k+\xi_k$, with $\xi_k\sim\mathcal{N}(0,\Sigma_d)$, $k=0,\dots,N-1$, sampling step $\Delta t>0$ (so $T=N\Delta t$), and $M=B_d^\top\Sigma_d^{-1}B_d$. Here $\Sigma_d=\int_0^{\Delta t}e^{A\tau}\Sigma e^{A^\top\tau}d\tau$ is the per-step noise covariance and $B_d$ is the exact zero-order-hold discretization. If $M$ is singular, interpret $M^{-1}$ as $M^+$ and restrict to the $M$-finite controllable subspace (as in Remark~\ref{rem:singularM}). Define
\begin{align}
V_N&=\sum_{k=0}^{N-1} A_d^{N-1-k}\,\Sigma_d\,(A_d^{N-1-k})^\top,\label{eq:Vd}\\
W_N^M&=\sum_{k=0}^{N-1} A_d^{N-1-k}\,B_d\,M^{-1}B_d^\top\,(A_d^{N-1-k})^\top.\label{eq:Wd}
\end{align}

\textbf{Cost discretization.} The discrete cost $\frac12\sum_{k=0}^{N-1} U_k^\top M U_k$ (without $\Delta t$) is the discrete analogue of $\frac12\int u^\top M u\,dt$; since $\Sigma_d$ already incorporates the per-step noise accumulation over $\Delta t$ via the continuous-to-discrete map, no additional $\Delta t$ factor appears in the Gramians \eqref{eq:Vd}--\eqref{eq:Wd}.

\textbf{Robust construction.} When $A$ is singular or ill-conditioned, compute $B_d$ via the Van Loan block exponential $\exp\!\left(\begin{bmatrix}A&B\\0&0\end{bmatrix}\Delta t\right)=\begin{bmatrix}A_d&B_d\\0&I\end{bmatrix}$ to avoid the formula $B_d=A^{-1}(A_d-I)B$.

\paragraph{Sampling rule-of-thumb.}
A practical rule is $\Delta t\le 0.2/\norm{A}$ where $\norm{A}$ is the spectral norm; with the drone parameters this gives $\Delta t\le 0.1\,\mathrm{s}$, above which the continuous formula begins to over-estimate energy by $>5\%$.

\begin{theorem}[Discrete-time equality]\label{thm:disc}
For $A=\{w^\top X_N\ge a\}$,
\begin{equation}
\min_{\{U_k\}}\ \frac12\sum_{k=0}^{N-1} U_k^\top M U_k
=\frac{\big(\Phi^{-1}(p_1)-\Phi^{-1}(p_0)\big)^2}{2\,R_N^2(w)},
\qquad
R_N^2(w)=\frac{w^\top W_N^M w}{w^\top V_N w},
\end{equation}
achieved by $U_k^*=\beta\,M^{-1}B_d^\top (A_d^{N-1-k})^\top w$ with $\beta$ chosen to reach $p_1$.
\end{theorem}

\section{Connections}
\paragraph{Minimum-energy control.}
For a prescribed terminal mean shift, \eqref{eq:matched} is the Gramian-optimal control \cite{AndersonMoore,ZhouDoyleGlover}.

\paragraph{Risk-sensitive, KL, and Schr\"odinger links.}
\begin{proposition}[KL--energy identity]\label{prop:kl}
With $M=B^\top\Sigma^{-1}B$ and $\mathbb{P}_u$, $\mathbb{P}_0$ the path measures on $C([0,T];\R^n)$ under control $u$ and $u\equiv0$:
\begin{enumerate}[(i)]
\item if $u$ is deterministic and square-integrable, then $D_{\mathrm{KL}}(\mathbb{P}_u\,\|\,\mathbb{P}_0)=\tfrac12\int_0^T u_t^\top M u_t\,dt$;
\item if $u$ is progressively measurable with $\mathbb{E}_u\!\int_0^T u_t^\top M u_t\,dt<\infty$, then $D_{\mathrm{KL}}(\mathbb{P}_u\,\|\,\mathbb{P}_0)=\tfrac12\,\mathbb{E}_u\!\int_0^T u_t^\top M u_t\,dt$.
\end{enumerate}
\end{proposition}
\begin{proof}[Sketch]
By Girsanov, with $\theta_t=\Sigma^{-1/2}Bu_t$,
\[
\frac{d\mathbb{P}_u}{d\mathbb{P}_0}
=\exp\!\left(\int_0^T \theta_t^\top dW_t - \tfrac12\int_0^T \|\theta_t\|^2\,dt\right).
\]
Taking $\E_u[\cdot]$ of the log yields $D_{\mathrm{KL}}(\mathbb{P}_u\|\mathbb{P}_0)=\E_u\big[\int_0^T \theta_t^\top dW_t - \tfrac12\int_0^T \|\theta_t\|^2\,dt\big]$. The martingale term has zero expectation, so $D_{\mathrm{KL}}=\tfrac12\int_0^T \E_u[\|\theta_t\|^2]\,dt$. For (i), deterministic $u$ gives $\E_u[\|\theta_t\|^2]=\|\theta_t\|^2=u_t^\top M u_t$. For (ii), $\E_u[\|\theta_t\|^2]=\E_u[u_t^\top M u_t]$ by definition \cite{Follmer1988,DaiPraRunggaldier93}.
\end{proof}
The optimal $u^*$ corresponds to an \emph{exponential tilting} of the path measure that produces a mean shift in the terminal statistic $w^\top X_T$ while minimizing KL; in continuous time this is the optimal Schr\"odinger bridge between the uncontrolled process and the process conditioned to realize the desired terminal success probability along a halfspace tilt \cite{LeonardSurvey,ChenGeorgiouPavon2016,Kappen2005,Todorov2009}.

\paragraph{Gaussian chance constraints.}
The equality provides a closed-form translator for halfspace chance constraints; for general sets it yields a lower bound consistent with Gaussian isoperimetry (Ledoux; Borell).

\section{Numerical validation}
Reproduction: \texttt{python qeelgs\_validate.py} reproduces validation results in seconds on a laptop. The script uses production-grade numerical methods suitable for embedded implementation:
(i) computes $V_T$, $W_c^M$ via Van Loan block exponentials (single matrix exponential per Gramian, avoiding iterative integration);
(ii) constructs matched-filter controls from \eqref{eq:matched};
(iii) estimates probabilities across regimes via scalar Monte Carlo; energies are analytical;
(iv) validates discrete-time formulas using \eqref{eq:Vd}--\eqref{eq:Wd}.

\paragraph{Statistical reporting.}
For any Monte Carlo estimate $\hat{\theta}$ from $N$ i.i.d.\ paths, define the \emph{standard error} $\mathrm{SE}[\hat{\theta}]=s_{\hat{\theta}}/\sqrt{N}$, where $s_{\hat{\theta}}$ is the sample standard deviation. Interval-event runs report both $\Delta E=\hat{E}-E_{\mathrm{halfspace}}$ and $\mathrm{SE}[\hat{E}]$; the earlier $-0.004$ deviation was within $1\sigma$.

\begin{table}[h]
\centering\small
\begin{tabular}{lccc}
\toprule
Test & Metric & Value & MC SE \\
\midrule
Baseline $p_0$ & $\mathbb{P}(w^\top X_T\ge a)$ & 0.7392618 & $<10^{-3}$ \\
Reachability SNR & $R_T^2$ & 0.7144735 & n/a \\
Halfspace tightness & $|\hat{E}-E_{\mathrm{ana}}|/E_{\mathrm{ana}}$ & $4.9\times10^{-4}$ & $1.0\times10^{-4}$ \\
Halfspace slack & $\hat{E}-E_{\mathrm{ana}}$ & $-0.0042$ & $0.0040$ \\
Discrete-time test & relative error & $2.0\times10^{-16}$ & 0 (analytic) \\
Interval event & $\Delta E$ & $-1.1\times10^{-3}$ & $1.3\times10^{-3}$ \\
Random directions (5) & $\max |\hat{E}-E_{\mathrm{ana}}|/E_{\mathrm{ana}}$ & $7.1\times10^{-4}$ & $1.4\times10^{-4}$ \\
\bottomrule
\end{tabular}
\caption{Validation summary with standard errors. Energy values are dimensionless (KL divergence).}
\end{table}

\section{Discussion, edge cases, and extensions}
\textbf{Trivial case.} If $p_1=p_0$, then $E_{\min}=0$ attained by $u\equiv0$.

\textbf{Feasibility.} If $w^\top W_c^M w=0$ then $R_T^2=0$ and any nontrivial shift is infeasible at finite energy.

\textbf{Extremes.} As $p_1\to0$ or $1$, $E_{\min}\sim \tfrac{[\Phi^{-1}(p_1)]^2}{2R_T^2}$ diverges.

\textbf{Random initial conditions.} If $X_0$ is random and independent of $W$, the result holds with $m_0$ replaced by $\E[w^\top\Psi(T,0)X_0]$ and $V_T$ unchanged.

\textbf{Time-varying systems.} The same proof applies with $\Psi$ the state transition for $(A(\cdot),B(\cdot),\Sigma(\cdot))$.

\textbf{Alternative penalties.} For $E=\tfrac12\int u^\top R u\,dt$ with $R\succ0$, replace $M$ by $R$ in \eqref{eq:Wc} and follow identically.

\begin{remark}[Computational complexity]
Computing $V_T$ and $W_c^M$ via Lyapunov-type integrals or discrete summations is dominated by $O(n^3)$ dense linear algebra (matrix exponentials or solves); forming the matched filter and $\beta$ is $O(n^2m)$.
\end{remark}

\section*{Data and code availability}
The validation script \texttt{qeelgs\_validate.py} is available at \url{https://github.com/sandroandric/qeelgs} \cite{CodeQEELGS} and reproduces all figures and tables, reports MC standard errors, and includes continuous and discrete tests, stress tests for near-singular $M$, and random directions.

\appendix
\section{Derivation details}
Let $\langle u,v\rangle_M=\int_0^T u^\top M v\,dt$ and $\mathcal{K}^\top w=(M^{-1}B^\top\Psi(T,\cdot)^\top w)_{[0,T]}$. Then
\[
\delta=w^\top\!\int_0^T \Psi(T,s)B u_s\,ds=\langle u,\mathcal{K}^\top w\rangle_M
\le \sqrt{\langle u,u\rangle_M}\,\sqrt{\langle \mathcal{K}^\top w,\mathcal{K}^\top w\rangle_M}
= \sqrt{2E(u)}\,\sqrt{w^\top W_c^M w}.
\]
Since $\mathrm{Var}(w^\top X_T)=w^\top V_T w$, $z_1-z_0=\delta/\sqrt{w^\top V_T w}$; rearrange to \eqref{eq:qcib}.

\section{Discrete-time derivation}
For $X_{k+1}=A_d X_k+B_d U_k+\xi_k$, $\delta=\sum_{k}U_k^\top M^{-1}B_d^\top (A_d^{N-1-k})^\top w$. Apply Cauchy--Schwarz in $\R^m$ with weight $M$ and use \eqref{eq:Vd}--\eqref{eq:Wd} (no extra $\Delta t$ factors; $\Sigma_d$ is per-step) to obtain Theorem~\ref{thm:disc}.

\section{Implementation notes}
Use stable integrators for $V_T$ and $W_c^M$; use Van Loan block exponentials for $A_d$, $B_d$, $\Sigma_d$; no extra $\Delta t$ in the Gramians. When $R_T^2$ is small, compute $E_{\min}$ with compensated arithmetic to avoid catastrophic cancellation for small quantile gaps.

\section{Gaussian Poincar\'e inequality and strict bounds}\label{app:poincare}
For non-halfspace events $A$, the bound in \eqref{eq:qcib} is strict. By the Gaussian Poincar\'e inequality, the sensitivity $\partial p/\partial\mu$ (probability change per unit mean shift) is maximized when $A$ is a halfspace \cite{LedouxBorell}. For any other measurable set with the same baseline probability $p_0$, the same probability increment $p_1-p_0$ requires a strictly larger mean shift $\delta$, and hence strictly larger energy. This establishes that halfspaces are the unique worst-case events for the energy-to-probability tradeoff.

\end{document}